\documentclass[12pt,amssymb]{article}
\usepackage{amssymb}
\usepackage{amsmath}
\usepackage{amsfonts}
\usepackage{array}
\usepackage{amsthm}

\def\be{\begin{equation}}
\def\ee{\end{equation}}
\def\bea{\begin{eqnarray}}
\def\eea{\end{eqnarray}}
\def\RR{\mathbb{R}}
\def\CC{\mathbb{C}}
\newtheorem{theorem}{Theorem}[section]

\numberwithin{equation}{section}

\begin{document}
\medskip

\begin{center}
{\Large \bf {On Fourier integral transforms }\\[0.5cm]
{for $q$-Fibonacci and $q$-Lucas polynomials}}

\bigskip

\noindent{Natig Atakishiyev and Pedro Franco}

\medskip

\noindent {Instituto de Matem\'aticas, Unidad Cuernavaca, Universidad Nacional \\
                Aut\'onoma de M\'exico, C.P. 62251 Cuernavaca, Morelos, M\'exico}
%\noindent{b) - Facultad de Ciencias, Universidad Aut\'onoma del Estado de  Morelos,\\
%C.P. 62250 Cuernavaca, Morelos, M\'exico}\\

\medskip
\noindent{E-mail:  natig@matcuer.unam.mx, pedro@matcuer.unam.mx}

\bigskip

\noindent{Decio Levi}

\medskip

\noindent {Dipartimento di Ingegneria Elettronica \\
Universit\`a degli Studi Roma Tre and  INFN Sezione di Roma Tre \\
Via della Vasca Navale 84, 00146 Roma, Italy}

\medskip
\noindent{E-mail:  levi@Roma3.infn.it}
\bigskip

\noindent{Orlando Ragnisco}

\medskip

\noindent {Dipartimento di Fisica \\
Universit\`a degli Studi Roma Tre and  INFN Sezione di Roma Tre \\
Via della Vasca Navale 84, 00146 Roma, Italy}

\medskip
\noindent{E-mail:  Ragnisco@Roma3.infn.it}
\medskip

\end{center}

\begin{abstract}

We study in detail two families of $q$-Fibonacci polynomials and $q$-Lucas
polynomials, which are defined by  non-conventional three-term recurrences.
They were recently introduced by Cigler and have been then employed by Cigler
and Zeng to construct novel $q$-extensions of classical Hermite polynomials.
We show that both of these $q$-polynomial families exhibit simple transformation
properties with respect to the classical Fourier integral transform.

\end{abstract}

\medskip



\medskip

\setcounter{equation}{0}
\section{Introduction}

The Askey scheme of hypergeometric orthogonal polynomials and their
$q$-analogues \cite{KSL} accumulates current knowledge about a large
number of these special functions. Depending on a number of parameters,
associated with each polynomial family, they occupy different levels
within the Askey hierarchy: for instance, the Hermite polynomials
$H_n(x)$ are on the ground level, the Laguerre and Charlier polynomials
$L_n^{(\alpha)}(x)$ and $C_n(x;a)$ are one level higher, and so on. All
polynomial families in this scheme are characterized by a ``canonical"
set of properties: they are solutions of differential or difference equations
of the second order, they can be generated by three-term recurrence relations,
they are orthogonal with respect to weight functions with finite or infinite
supports, they obey Rodrigues-type formulas, and so on. Of course, many other
polynomial families of interest arise both in pure and applied mathematics,
which do not belong to the Askey $q$-scheme only because they lack some of
the above mentioned characteristics properties. So this paper is aimed at
exploring in detail two particular $q$-polynomial families of this type,
namely, $q$-Fibonacci and $q$-Lucas polynomials, which are defined by
non-conventional three-term recurrences. They were introduced in
\cite{Cigl-I}--\,\cite{Cigl-III} and have been studied in detail
in \cite{Cigl-Zeng}. The Fibonacci and Lucas sequences and polynomials
have many physical applications; for example they appear in the study of diatomic
chains \cite{lang}, in dynamical systems and chaos theory \cite{sch},
in Ising models \cite{ising}, etc.

Our main result is to show that both of these $q$-polynomials exhibit simple
transformation properties with respect to the classical Fourier integral
transform.

Throughout this exposition we employ standard notations of the theory
of special functions (see, for example, \cite{GR}--\,\cite{KSL}).
In sections 2 and 3 we present some basic background facts about Fibonacci
and Lucas polynomials and their $q$-extensions, respectively, which are then
used in section 4 in order to find explicit forms of Fourier integral transforms
for the $q$-Fibonacci and $q$-Lucas polynomials. Section 5 contains the conclusions
and a brief discussion of some further research directions of interest. Finally,
Appendix concludes this work with the derivation of two transformation formulas
for hypergeometric ${}_2F_1$-polynomials, associated with the Chebyshev polynomials
$T_n(x)$ and $U_n(x)$.

%\setcounter{equation}{0}
%%%%%%%%%%%%%%%%%%%%%%%%%%
%%%%%%%%%%%%%%%%%%%%%%%%%%
\section{Fibonacci and Lucas polynomials}
%%%%%%%%%%%%%%%%%%%%%%%%%%
%%%%%%%%%%%%%%%%%%%%%%%%%%

In this Section we review the basic well known facts about the Fibonacci and Lucas polynomials.

%%%%%%%%%%%%%%%%%%%%%%%%%
\subsection{The Fibonacci polynomials }
%%%%%%%%%%%%%%%%%%%%%%%%%

$F_{n}(x,s)$ are defined by the
three-term recurrence relation
\be
F_{n+1}(x,s) = x\,F_{n}(x,s)\,+\,s\,F_{n-1}(x,s)\,,\qquad\quad
n\geq1\,,                                                  \ee
with initial values $F_0(x,s)=0$ and $F_1(x,s)=1$ \cite{Cigl-I,Cigl-II}.
They are also given by the explicit sum formula
\bea
F_{n+1}(x,s)&=&\sum_{k=0}^{\lfloor\,n/2\,\rfloor}{n-k
\atopwithdelims()k}s^{\,k}\,x^{\,n-\,2k}  \nonumber \\
&=&x^{\,n}\,{}_2F_1\Big(-\frac{n}{2}\,,\frac{1-n}{2}
\,;-\,n\,\Big|-4s/x^2\,\Big)\,,\qquad n\geq0\,, \eea
where ${\,n\,\atopwithdelims()\,k\,}= n!/k!(n-k)!$ is a binomial
coefficient and $\lfloor x \rfloor$ denotes the greatest integer
in $x$. The Fibonacci polynomials $F_{n}(x,s)$ have the following
generating function
\be
f_F(x,s;t):=\,\sum_{n=0}^{\infty}\,F_{n}(x,s)\,t^{\,n}\,=\,\frac
{t}{1-x\,t-s\,t^{\,2}}\,,\qquad\qquad |\,t\,|<1\,,           \ee
easily derived through the three-term recurrence
relation  (2.1).

The Fibonacci polynomials $F_{n}(x,s)$ are normalized so that
$F_n(x,1)=f_{n}(x)$ (where $f_{n}(x)$ are the Fibonacci polynomials
introduced by Catalan ( see \cite{Koshy}, formula (37.1) on p.443)
and for the particular values of $x=s=1$ the recursion (2.1) generates
a classical sequence of the Fibonacci numbers $\{F_{n}\}_{n=1}^{\infty}
\equiv\{1,1,2,3,5,8,13,...\}$ and the relation (2.3) reduces to the
well-known generating function $f_F(1,1;t)$ for the Fibonacci numbers
$\{F_n\}$, which "have been a source of delight to professional and
amateur mathematicians for seven centuries" \cite{Bers} (see also
\cite{Dunlap, NalliHauk, StakhAran}).

We call attention to the fact that the Fibonacci polynomials (2.2)
(of degree $n$ in $x$ and of $\lfloor\,n/2\,\rfloor$ in $s$) can
also be represented as
\bea
F_{n+1}(x,s)&=&\Big(2\,\sqrt{s}\,\Big)^{\,n}\,p_{\,n}
^{\,(F)}\Bigg(\frac{x}{2\,\sqrt{s}}\Bigg)\,,\nonumber \\
p_n^{\,(F)}(x)&:=&x^{\,n}\,{}_2F_1\Big(-\frac{n}{2}
\,,\frac{1-n}{2}\,;-\,n\,\Big|-1/x^2\,\Big)\,,\eea
so that the fundamental properties of $F_{n}(x,s)$ are basically defined 
by the {\it monic} polynomials $p_{\,n}^{\,(F)}(x)$.\footnote{We recall 
that an arbitrary polynomial $p_n(x)=\sum_{k=0}^n c_{n,\,k}\,x^k$ of degree 
$n$ can be written in the {\it monic form}  $p_n^{(M)}(x)=c_{n,\,n}^{-1}\,
p_n(x)=x^n + c_{n,\,n}^{-1}\sum_{k=0}^{n-1}c_{n,\,k}\,x^k$ just by changing 
its normalization.} Moreover, it turns out that the polynomials 
$p_{\,n}^{\,(F)}(x)$ {\it are essentially} the Chebyshev polynomials of 
the second kind $U_n(z)$ in an imaginary argument $z\in\CC$ (see p.449 in 
\cite{Koshy}). Indeed, recall that the Chebyshev polynomials $U_n(z)$
have an explicit representation (see formula (23) on p.185 in \cite{HTF})
\be
U_{n}(z)\,=\,\sum_{k=0}^{\lfloor\,n/2\,\rfloor}(-1)
^k\,{n-k\atopwithdelims()k}\,(2z)^{\,n-\,2k}\,.\ee
Therefore from (2.2), (2.4) and (2.5) one readily deduces that
\be
F_{n+1}(x,s)\,=\,\Big(\!-{\rm i}\sqrt{s}\,\Big)
^{n}\,U_n\Big({\rm i}\,x/2\sqrt{s}\Big)    \ee
and, consequently, $p_{\,n}^{\,(F)}(x)=({-\rm i}/2)^n\,U_n\Big
({\rm i}\,x\Big)$. Observe that from the three-term recurrence
relation
\be \label{pippo1}
2z\,U_{n}(z)\,= \,U_{n+1}(z)\,+\,U_{n-1}(z)
\ee
for the Chebyshev polynomials of the second kind $U_n(z)$ it follows
at once that
\be
p_{\,n+1}^{\,(F)}(x) = x\,p_{\,n}^{\,(F)}(x)\,+\,\frac14\,p_{\,n-1}
^{\,(F)}(x)\,,\qquad\quad n\geq1\,,                             \ee
with initial values $p_{\,0}^{\,(F)}(x)=1$ and $p_{\,1}^{\,(F)}(x)=x$.
The coefficients in the three-term recurrence relations (2.8) are $A_n=1$
and $C_n=-1/4$; so that they do not satisfy the conditions $A_n\,C_{n+1}>0$
of Favard's characterization theorem (see, for example, (7.1.5) on p.175 in
\cite{GR}). This means that there is not a unique positive orthogonality measure
for the polynomials $p_{\,n}^{\,(F)}(x)$.

Let us remark that:
\begin{enumerate}
\item Since the Chebyshev polynomials of the second kind
$U_n(z)$ can be expressed in terms of the hypergeometric ${}_2F_1$
polynomials as (see, for example, (9.8.36) in \cite{KSL})
\be
U_n(z)\,=\,(n+1)\,{}_2F_1\Big(-n\,,n+2\,;\,3/2\,\Big|\frac
{1-z}{2}\,\Big)\,,                                      \ee
then (2.6) is consistent with the second line in (2.2) only if a
transformation formula
\be
(n+1)\,{}_2F_1\Big(-n\,,n+2\,;\,3/2\,\Big|\frac{1-z}{2}
\,\Big)\,=\,(2 z)^{\,n}\,{}_2F_1\Bigg(-\frac{n}{2}\,,
\frac{1-n}{2}\,;-\,n\,\Big|1/z^{\,2}\,\Big)        \ee
is valid. A direct proof of  (2.10) is given  in the Appendix.

\item Since the Chebyshev polynomials of the second kind $U_n(x)$
are known to satisfy the second order differential equation (see (9.8.44)
in \cite{KSL})
\be
\Big[(1-x^2) \frac{d^{\,2}}{dx^{\,2}}-3\,x\frac{d}
{dx}+n(n+2)\Big]\,U_n(x)\,=\,0\,,              \ee
one readily deduces that
\be
\Big[\Big(1+x^2\Big)\frac{d^{\,2}}{dx^{\,2}} + 3\,x\frac{d}{dx}
\Big]\,p_{\,n}^{\,(F)}(x)\,=\,n(n+2)\,p_{\,n}^{\,(F)}(x)\,.\ee

\item The generating function for the Chebyshev polynomials of the second
kind $U_n(x)$ is known to be of the form (see (9.8.56) in \cite{KSL})
\be
\sum_{n=0}^{\infty}\,t^{\,n}\,U_{n}(x)\,=\,\frac{1}{1-2\,x\,t
+\,t^{\,2}}\,,\qquad\qquad |\,t\,|<1\,,                   \ee
so that combining (2.6) with (2.13) we get (2.3).
\end{enumerate}

%%%%%%%%%%%%%%%%%%%%%%%%
\subsection{ The Lucas polynomials}
%%%%%%%%%%%%%%%%%%%%%%%%

$L_{n}(x,s)$ for $n\geq3$ ($L_0(x,s)=1$) are defined by the same
three-term recurrence relation as in (2.1), but with initial values
$L_1(x,s)=x$ and $L_2(x,s)=x^2+2s$ \,\cite{Cigl-Zeng}. They have the
explicit sum formula
\bea
L_{n}(x,s)&=&\sum_{k=0}^{\lfloor\,n/2\,\rfloor}\,\frac{n}{n-k}
\,{n-k\atopwithdelims()k}\,s^{\,k}\,x^{\,n-\,2k} \nonumber  \\
&=& x^{\,n}\,{}_2F_1\Bigg(-\frac{n}{2}\,,\frac{1-n}{2}\,;1-n
\,;-\frac{4s}{x^2}\,\Bigg),\quad\quad n\geq0\,.               \eea
The Lucas polynomials $L_{n}(x,s)$ have a generating function
of the form
\be
f_L(x,s;t):=\,\sum_{n=0}^{\infty}\,L_{n}(x,s)\,t^{\,n}\,=\,\frac
{1+s\,t^{\,2}}{1-x\,t-s\,t^{\,2}}\,,\qquad\qquad |\,t\,|<1\,.\ee
To derive (2.15), multiply both sides of the three-term recurrence
relation for $L_{n}(x,s)$  by the factor $t^{\,n+1}$ and sum all
three terms with respect to the index $n$ from $n=2$ to infinity by
taking into account initial values of $L_0(x,s)$, $L_1(x,s)$ and
$L_2(x,s)$.

The Lucas polynomials $L_{n}(x,s)$ are normalized in such a way that
$L_{n}(x,1)=l_{n}(x)$ (where $l_{n}(x)$ are the Lucas polynomials
studied by Bicknell, see p.459 in \cite{Koshy}) and for the particular
values of $x=s=1$ the sequence $\{L_n(1,1)\}_{n=1}^{\infty}$ reproduces
Lucas numbers $\{L_n\}\equiv\{1,3,4,7,11,18,...\}$ and the relation (2.15)
reduces to the generating function $f_L(1,1;t)$ for these numbers
$\{L_n\}$ (see \cite{StakhAran} for applications of Lucas numbers
and \cite{NalliHauk} for some generalizations of Lucas polynomials).

We call attention to the fact that the Lucas polynomials (2.14) (of
degree $n$ in $x$ and of degree $\lfloor\,n/2\,\rfloor$ in $s$) can
also be represented as
\bea
L_{n}(x,s)&=&\,s^{\,n/2}\,p_{\,n}^{\,(L)}\Bigg(\frac{x}
{\sqrt{s}}\Bigg)\,,                          \nonumber \\
p_n^{\,(L)}(x)&:=&x^{\,n}\,{}_2F_1\Bigg(-\frac{n}{2}\,,\frac
{1-n}{2}\,;1-\,n\,\Bigg|-\frac{4}{x^2}\,\Bigg)\,,       \eea
so that the fundamental properties of $L_{n}(x,s)$ are basically
defined by the monic polynomials $p_n^{\,(L)}(x)$. Moreover, it
turns out that the polynomials $p_n^{\,(L)}(x)$  are in fact the
Chebyshev polynomials of the first kind $T_n(z)$ in an imaginary
argument $z$. Indeed, recall that the Chebyshev polynomials
$T_n(z)$ have an explicit representation (see (23) on p.185 in
\cite{HTF})
\bea
T_{n}(z)&=&\frac{n}{2}\,\sum_{k=0}^{\lfloor\,n/2\,\rfloor}\frac
{(-1)^k\,(n-k-1)!}{k!\,(n-2k)!}\,(2z)^{\,n-\,2k}    \nonumber \\
&\equiv & 2^{n-1}\,z^{\,n}\,{}_2F_1\Big(-\frac{n}{2}\,,\frac{1-n}
{2}\,;1-\,n\,\Big|\,1/z^{2}\,\Big)\,, \qquad n\geq1\,.       \eea
Therefore from (2.14), (2.16) and (2.17) it follows that
\be
L_{0}(x,s)=1\,,\qquad L_{n}(x,s)\,=\,2\Big(\!-{\rm i}\sqrt{s}\,\Big)^{n}
\,T_n\Bigg(\frac{{\rm i}\,x}{2\sqrt{s}}\Bigg)\,,\quad n\geq1\,,  \ee
and, consequently, $p_{\,0}^{\,(L)}(x)=1\,,\,p_{\,n}^{\,(L)}(x)=2
({-\rm i})^n\,T_n\Big ({\rm i}\,x/2\Big)\,,\,n\geq1$. Observe
that from (2.16), (2.18) and the three-term recurrence relation (\ref{pippo1})
%\be
%2z\,T_{n}(z)\,= \,T_{n+1}(z)\,+\,T_{n-1}(z)
%\ee
for the Chebyshev polynomials of the first kind $T_n(z)$ it
follows at once that
\be
p_{\,n+1}^{\,(L)}(x) = x\,p_{\,n}^{\,(L)}(x)\,+\,p_{\,n-1}
^{\,(L)}(x)\,,\qquad\quad n\geq1\,,                    \ee
with initial values $p_{\,0}^{\,(L)}(x)=1$ and $p_{\,1}^{\,(L)}
(x)=x$. The coefficients in the three-term recurrence relations
(2.20) are $A_n=-\,C_n=1$; so that they do not satisfy the conditions
$A_n\,C_{n+1}>0$ of Favard's characterization theorem (see, for
example, (7.1.5) on p.175 in \cite{GR}).

It is to be stressed that there are at least {\it three} direct
consequences of the connection (2.18) between the Lucas polynomials
$L_{n}(x,s)$ and the Chebyshev polynomials $T_n\Big({\rm i}\,x/2
\sqrt{s}\Big)$ of the first kind.

\begin{enumerate}
\item Since the Chebyshev polynomials of the first kind
$T_n(z)$ can be written in terms of the hypergeometric ${}_2F_1$
polynomials as (see, for example, (9.8.35) in \cite{KSL})
\be
T_n(z)\,=\,{}_2F_1\Big(-n\,,n\,;\,1/2\,\Big|\frac{1-z}{2}\,\Big)\,,
\ee
then (2.18) is consistent with the second line in (2.14) only if a
transformation formula
\be
{}_2F_1\Big(-n\,,n\,;\,1/2\,\Big|\frac{1-z}{2}\,\Big)
\,=\,2^{n-1}\,z^{\,n}\,{}_2F_1\Big(-\frac{n}{2}\,,\frac
{1-n}{2}\,;1-\,n\,\Big|\,1/z^{2}\,\Big)\,,\quad n\geq1\,,
\ee
is valid. A proof of this identity  is given in Appendix.

\item Since the Chebyshev polynomials of the first kind
$T_n(x)$ are known to satisfy the second order differential equation
(see (9.8.43) in \cite{KSL})
\be
\Big[(1-x^2) \frac{d^{\,2}}{dx^{\,2}}- x \frac{d}{dx}+ n^{2}
\Big]\,T_n(x)\,=\,0\,,                                  \ee
from the relation $p_{\,n}^{\,(L)}(x)=2({-\rm i})^n\,T_n\Big ({\rm i}
\,x/2\Big)$ it follows at once that
\be
\Bigg[\Big(4+x^2\Big)\frac{d^{\,2}}{dx^{\,2}}+ x \frac{d}{dx}\Bigg]
\,p_{\,n}^{\,(L)}(x)\,=\,n^{\,2}\,p_{\,n}^{\,(L)}(x)\,.        \ee

\item The generating function for the Chebyshev polynomials
of the first kind $T_n(x)$ is known to be of the form (see (9.8.50)
in \cite{KSL})
\be
\sum_{n=0}^{\infty}\,t^{\,n}\,T_{n}(x)\,=\,\frac{1-x\,t}
{1-2\,x\,t+\,t^{\,2}}\,,\qquad\qquad |\,t\,|<1\,,   \ee
so that if one combines (2.18) with (2.25), this leads to  (2.15).
\end{enumerate}

%\setcounter{equation}{0}
%%%%%%%%%%%%%%%%%%%%%%%%%%%%
%%%%%%%%%%%%%%%%%%%%%%%%%%%%
\section{$q$-Fibonacci and $q$-Lucas polynomials}
%%%%%%%%%%%%%%%%%%%%%%%%%%%%
%%%%%%%%%%%%%%%%%%%%%%%%%%%%

In this Section we review mainly tfacts about the $q$-Fibonacci and $q$-Lucas polynomials.

%%%%%%%%%%%%%%%%%%%%%%%%%%%%
\subsection{ $q$-Fibonacci polynomials.} 
%%%%%%%%%%%%%%%%%%%%%%%%%%%%

Cigler defined in
\cite{Cigl-I,Cigl-II} a novel $q$-analogue of Fibonacci polynomials
$F_{n}(x,s)$, which satisfy the rather non-standard three-term recursion
\be
F_{n+1}(x,s|\,q) = \Big[\,x + (q-1)s\,{\cal D}_q\,\Big]\,F_{n}
(x,s|\,q)+ s\,F_{n-1}(x,s|\,q)\,,\quad\quad n\geq1\,,      \ee
where the initial values are $F_{0}(x,s|\,q)=0$ and $F_{1}(x,s|\,q)=1$,
and the Hahn $q$-difference operator ${\cal D}_q$ is defined as
\be
{\cal D}_q\,f(x):=\frac{f(x)-f(qx)}{(1-q)x}\,.      \ee
The $q$-Fibonacci polynomials $F_n(x,s|\,q)$ are explicitly given
in the form
\bea \label{pippo2}
F_{n+1}(x,s|\,q)= \sum_{k=0}^{\lfloor\,n/2\,\rfloor\,}\,q^{k(k+1)
/2}\,{n-k\atopwithdelims[]k }_q \,s^{\,k}\,x^{\,n-\,2\,k} , \\ \label{pippo3}
=x^{\,n}{}_4\phi_1\Bigg(q^{-\,n/2},q^{(1-n)/2},-\,q^{-\,n/2}
,-\,q^{(1-n)/2}\,;q^{-\,n}\,\Bigg|\,q\,;-\frac{q^{\,n}s}
{x^2}\Bigg)\,,\quad\quad n\geq0\,,                  \eea
where ${\,n\,\atopwithdelims []\,k\,}_q$ stands for the $q$-binomial
coefficient,
\be
{\,n\,\atopwithdelims []\,k\,}_q :=\frac{(q;q)_n}
{(q;q)_k(q;q)_{n-k}}\,,        \label{pippo3a}               \ee
and $(z;q)_n$ is the $q$-shifted factorial, that is, $(z;q)_0=1$,
\,$\,(z;q)_n=\prod_{k=0}^{n-1}(1-zq^k)$ for $n\geq 1$\,.
%and we have employed the conventional notation $(a_1,a_2,...,a_k;q)
%_n:=\prod_{j=1} ^k(a_j;q)_n$ for products of $q$-shifted factorials.

A $q$-extension of the generating function $f_F(x,s;t)$, associated
with the $q$-Fibonacci polynomials $F_n(x,s|\,q)$, has the form
({\it cf.} (2.3))
\bea
f_F(x,s;t|\,q)&:=& \,\sum_{n=0}^{\infty}\,F_{n}(x,s|\,q)\,t^{\,n}
\,=\,\frac{t}{1-x\,t}\,{}_1\phi_1\Big(q\,;qxt\Big|\,q\,;-\,qst^2\Big)
\nonumber \\
&=& t\,{}_2\phi_1\Big(\!-\frac{qst}{x}\,,q\,;0\,\Big|\,q\,;xt\Big)
\,,\qquad \qquad |\,t\,|<1\,.     \label{pippo4}                           \eea
The $q$-Fibonacci polynomials $F_n(x,s|\,q)$ are defined in such a way
that in the limit as $q\to1$ they reduce to the polynomials $F_n(x,s)$,
\be \label{pippo5}
F_n(x,s|\,1)\,\equiv\,\lim_{q\to1}F_n(x,s|\,q)\,=\,F_n(x,s)\,,\ee
and the generating function (\ref{pippo4}) for $F_n(x,s|\,q)$ coincides
therefore in this limit with (2.3), associated with the polynomials
$F_n(x,s)$.
The appearance in (\ref{pippo4}) of the two equivalent expressions in terms of
the either ${}_1\phi_1$\,, or ${}_2\phi_1$ basic hypergeometric
functions is fully consistent with the limit case of Heine's
transformation formula (see \cite{KSL}, formula (1.13.13) on p.20)
\bea \label{pippo6}
\,{}_2\phi_1\Big(a\,,b\,;0\Big|\,q\,;z\Big)\,=\,\frac{(bz;q)_{\infty}}
{(z;q)_{\infty}}\,{}_1\phi_1\Big(b\,;bz\Big|\,q\,;az\Big),         \eea
where $a=-\,qst/x$, $b=q$ and $z=xt$ (\,which means that the
quotient $(bz;q)_{\infty}/(z;q)_{\infty}$ simply reduces in
this case to the factor $1/(1-z)=1/(1-xt)$\,)\,.

%%%%%%%%%%%%%%%%%%%%%%%%%%%
\subsection{ $q$-Lucas polynomials.}
%%%%%%%%%%%%%%%%%%%%%%%%%%%

 A $q$-extension of the Lucas
polynomials $L_{n}(x,s)$ was introduced by Cigler in \cite{Cigl-II,Cigl-III}
as
\be \label{pippo7}
L_{n}(x,s\,|\,q)\,=\,L_n\Big(\,x +(q-1)s\,{\cal D}
_q\,,s\Big)\cdot1\,,                          \ee
where the Hahn $q$-difference operator ${\cal D}_q$ is defined in (3.2).
From (\ref{pippo7}) and the three-term recurrence relation for $L_{n}(x,s)$ it then
follows that the so introduced $q$-Lucas polynomials satisfy a non-standard
three-term recurrence relation of the form of (3.1) for $L_{n}
(x,s\,|\,q)$ for $ n\geq2\,$.
The initial values in this case are
$L_{0}(x,s|\,q)=1$, $L_{1}(x,s\,|\,q)=x$ and
$L_{2}(x,s\,|\,q)= x^2+(1+q)s$.
The $q$-Lucas polynomials $L_n(x,s\,|\,q)$ have an explicit
sum formula \\({\it cf.} (2.14) and (\ref{pippo2}))
\bea \label{pippo8}
L_{n}(x,s\,|\,q)= \sum_{k=0}^{\lfloor\,n/2\,\rfloor}
\,q^{k(k-1)/2}\frac{[n]_q}{[n-k]_q}\,{n-k\atopwithdelims[]k}
_q \,s^{\,k}\,x^{\,n-\,2\,k}                       \\ \label{pippo9}
= x^{\,n}{}_4\phi_1\Bigg(q^{-\,n/2},q^{(1-n)/2},-\,q^{-\,n/2}
,-\,q^{(1-n)/2}\,;q^{1-\,n}\,\Bigg|\,q\,;-\frac{q^{\,n}s}{x^2}
\Bigg)\,,                                                 \eea
where $[n]_q:=(1-q^n)/(1-q)$.
%and we have employed the conventional notation $(a_1,a_2,...,a_k;q)_n
%:=\prod_{j=1} ^k(a_j;q)_n$ for products of $q$-shifted factorials.

A $q$-extension of the generating function $f_L(x,s;t)$, associated
with the $q$-Lucas polynomials $L_n(x,s\,|\,q)$, has the form
({\it cf.} (2.15))
\bea
f_L(x,s;t|\,q)&:=&\,\sum_{n=0}^{\infty}\,L_{n}(x,s\,|\,q)\,t^{\,n}
\,=\,\frac{1+st^2}{1-x\,t}\,{}_1\phi_1\Big(q\,;qxt\Big|\,q\,;
-\,qst^2\Big)                                       \nonumber\\
&=&(1+st^2)\,{}_2\phi_1\Big(\!-\frac{qst}{x}\,,q\,;0\,\Big|\,q
\,;xt\Big)\,,\qquad\qquad|\,t\,|<1\,.       \label{pippo10}              \eea
In the limit as $q\to1$ the $q$-Lucas polynomials $L_n(x,s\,|\,q)$
reduce to the polynomials $L_n(x,s)$, given in (2.14), and the
generating function (\ref{pippo10}) for $L_n(x,s\,|\,q)$ coincides in this
limit with (2.15).
The appearance of the two equivalent expressions in (\ref{pippo10}) in
terms of the either ${}_1\phi_1$\,, or ${}_2\phi_1$ basic
hypergeometric functions is, as before, fully consistent with a limit
case of Heine's transformation formula (\ref{pippo6}).

Since the polynomials  $F_{n}(x,s)$ and $L_{n}(x,s)$ satisfy
the same recurrence relation (2.1) but with different initial conditions,
they are known to be interconnected by the relation  $L_{n}(x,s)=F_{n+1}
(x,s)+s\,F_{n-1}(x,s)$ ({\it cf.} the classical relation $2T_n(x)=U_n(x)
-U_{n-2}(x)$ for the Chebyshev polynomials). Similarly, a $q$-extension 
of this relation,
\be \label{pippo11}
L_{n}(x,s|\,q)=F_{n+1}(x,s|\,q)+s\,F_{n-1}(x,s|\,q)\,,
\ee
interconnects two $q$-polynomial families $F_{n}(x,s|\,q)$ and $L_{n}(x,s|\,q)$.
The relation (\ref{pippo11}) is readily verified by using the explicit forms (\ref{pippo3}) and
(\ref{pippo9}) of the polynomials $F_{n}(x,s|\,q)$ and $L_{n}(x,s|\,q)$ and an identity
$$
{n-k\atopwithdelims[]k }_q\,=\,\frac{1-q^{n-k}}{1-q^{k}}
\,{n-k-1\atopwithdelims[]k -1}_q                      $$
for the $q$-binomial coefficient (\ref{pippo3a}).

%\setcounter{equation}{0}
%%%%%%%%%%%%%%%%%%%%%%%%%%%
%%%%%%%%%%%%%%%%%%%%%%%%%%%
\section{Fourier transforms of $F_n(x,s|\,q)$ and $L_n(x,s|\,q)$}
%%%%%%%%%%%%%%%%%%%%%%%%%%%
%%%%%%%%%%%%%%%%%%%%%%%%%%%
In this section we derive
explicit formulas of the classical Fourier integral transform for the
$q$-Fibonacci and $q$-Lucas polynomials $F_n(x,s|\,q)$ and $L_n(x,s|\,q)$.

%%%%%%%%%%%%%%%%%%%%%%%%%%
\subsection{$q$-Fibonacci} 
%%%%%%%%%%%%%%%%%%%%%%%%%%
Let us consider the $F_n(x,s|\,q)$ polynomials. 
 To that end  let us first
define how this $q$-polynomial family changes under the transformation
$q\to 1/q$. If one rewrites the defining sum formulas (\ref{pippo2}) for
$F_n(x,s|\,q)$ as
\be \label{pippo4.1}
F_{n+1}(x,s|\,q)= \sum_{k=0}^{\lfloor\,n/2\,\rfloor}\,c_{n,\,k}
^{(F)}(q)\,s^{\,k}\,x^{\,n-\,2\,k}\,,                      \ee
then the coefficients in (\ref{pippo4.1}) are
\be \label{pippo4.2}
c_{n,\,k}^{(F)}(q):= \,q^{k(k+1)/2}\,{n-k\atopwithdelims[]k }_q\,.\ee
From definition of the $q$-binomial coefficient ${\,n\,\atopwithdelims[]
\,k\,}_q$ in (\ref{pippo3a}) it is not hard to derive an inversion formula
\be
{\,n\,\atopwithdelims[]\,k\,}_{1/q}\,=\,q^{k(k-\,n)}
\,{\,n\,\atopwithdelims[]\,k\,}_q                \ee
with respect to the change $q\to 1/q$. Consequently, from (4.2) and
(4.3) it follows at once that
\be
c_{n,\,k}^{(F)}\Big(q^{-1}\Big)\,=\,q^{k(k-n-1)}\,c_{n,\,k}^{(F)}(q)\,.
\ee
This means that ({\it cf.} (\ref{pippo3}))
$$
F_{n+1}\Big(x,s\Big|\,q^{-1}\Big)\,\equiv\,\sum_{k=0}^{\lfloor\,n/2\,
\rfloor}\,c_{n,\,k}^{(F)}(q^{-1})\,s^{\,k}\,x^{\,n-\,2\,k}=\sum_{k=0}
^{\lfloor\,n/2\,\rfloor}\,q^{k(k-n-1)}\,c_{n,\,k}^{(F)}(q)\,s^{\,k}
\,x^{\,n-\,2\,k}                                                 $$
\be=x^{\,n}{}_4\phi_3\Bigg(q^{-\,n/2},q^{(1-n)/2},-\,q^{-\,n/2},-\,
q^{(1-n)/2}\,;q^{-\,n}\,,0\,,0\Bigg|\,q\,;-\frac{s}{q\,x^{\,2}}\Bigg)\,.
\ee
Take into account the well-known Fourier transform
$$
\int_{\RR}\,e^{\,{\rm i}xy -\,x^2/2}\,dx\,=\,\sqrt{2\pi}\,e^{-\,y^2/2} $$
for the Gauss exponential function $e^{-\,x^2/2}$ and computing the Fourier
integral transform of the exponential function $\exp\,[\,{\rm i}(n-2k)\,
\kappa\,x - x^2/2\,]$, we get
\bea
\int_{\RR}\,e^{\,{\rm i}xy\,+ \,{\rm i}\,(n-\,2k)\,\kappa\,x-\,x^2/2}
\,dx&=&\sqrt{2\pi}\,e^{-\,[\,y+(n-\,2k)\,\kappa\,]^2/2} \nonumber \\
&=&\sqrt{2\pi}\,q^{\,n^2/4}\,q^{\,k(k-n)}\,e^{-\,(n-\,2k)
\,\kappa\,y\,-\,y^2/2}\,,                           \eea
where $q=e^{-\,2{\kappa}^2}$.

We are now in a position to formulate and prove the following theorem.
\begin{theorem}
The classical Fourier integral transform of the $q$-Fibonacci polynomials
$F_{n+1}(a\,e^{\,{\rm i}\,\kappa\,x},s|\,q)$ times the Gauss exponential
function $e^{-\,x^2/2}$ has the form:
\be
\int_{\RR}\,F_{n+1}\Big(a\,e^{\,{\rm i}\,\kappa\,x},s\Big|\,q\Big)\,e
^{\,{\rm i}\,x\,y\,-\,x^2/2}\,dx\,=\,\sqrt{2\pi}\,q^{\,n^2/4}\,F_{n+1}
\Big(a\,e^{-\kappa\,y},qs\Big|\,q^{-1}\Big)\,e^{-\,y^2/2}\,,      \ee
where $a$ is an arbitrary constant factor.
\end{theorem}
\begin{proof}
Using (\ref{pippo2}), the Fourier
integral transform (4.6) and  the interrelation (4.4) between
the coefficients $c_{n,\,k}^{(F)}(q)$ and $c_{n,\,k}^{(F)}(q^{-1})$, we get
%\newpage
\bea
& & \int_{\RR}\,F_{n+1}\Big(a\,e^{\,{\rm i}\kappa\,x},s\Big|\,q\Big)
\,e^{\,{\rm i}\,x\,y\,-\,x^2/2}\,dx                   \nonumber   \\
& &=\sum_{k=0}^{\lfloor\,n/2\,\rfloor}\,c_{n,\,k}^{(F)}(q)\,s^{\,k}\,
a^{\,n-\,2\,k}\,\int_{\RR}\,e^{\,{\rm i}\,x\,y\,+ \,{\rm i}\,(n-\,2k)
\,\kappa\,x-\,x^2/2}\,dx                                 \nonumber \\
& & =\,\sqrt{2\pi}\,q^{\,n^2/4}\,e^{-\,y^2/2}\,\sum_{k=0}^{\lfloor
\,n/2\,\rfloor}\,q^{k(k-n)}\,c_{n,\,k}^{(F)}(q)\,s^{\,k}\,\Big
(a\,e^{-\,\kappa\,y}\Big)^{\,n-\,2\,k}             \nonumber \\
& & =\,\sqrt{2\pi}\,q^{\,n^2/4}\,e^{-\,y^2/2}\,\sum_{k=0}^{\lfloor
\,n/2\,\rfloor}\,c_{n,\,k}^{(F)}(q^{-1})\,(q\,s)^{\,k}\,\Big(a\,e^
{-\,\kappa\,y}\Big)^{\,n-\,2\,k}                      \nonumber \\
& &=\,\sqrt{2\pi}\,q^{\,n^2/4}\,F_{n+1}\Big(a\,e^{-\,\kappa\,y}
,qs\Big|\,q^{-1}\Big)\,e^{-\,y^2/2}\,.      \nonumber    \eea
\end{proof}

\subsection{$q$-Lucas} We turn now to determine an explicit form
of classical Fourier integral transform for the $q$-Lucas polynomials
$L_n(x,s\,|\,q)$. If one rewrites the defining sum formula (\ref{pippo8}) for
$L_n(x,s\,|\,q)$ as
\be
L_{n}(x,s\,|\,q)= \sum_{k=0}^{\lfloor\,n/2\,\rfloor}
\,c_{n,\,k}^{(L)}(q)\,s^{\,k}\,x^{\,n-\,2\,k}\,,\ee
then the coefficients in (4.8) are
\be
c_{n,\,k}^{(L)}(q):= \,q^{k(k-1)/2}\frac{[n]_q}{[n-k]_q}
\,{n-k\atopwithdelims[]k }_q\,.                     \ee
From the starting definition of the symbol $[n]_q$ in (\ref{pippo8}) it is not
hard to show that
\be
[n]_{1/q}\,=\,q^{1-\,n}\,[n]_q\,.     \ee
Consequently, from (4.3) and (4.10) it follows at once that
\be
c_{n,\,k}^{(L)}\Big(q^{-1}\Big)\,=\,q^{k(k-n)}
\,c_{n,\,k}^{(L)}(q)\,.                    \ee
This means that ({\it cf.} (\ref{pippo8}, \ref{pippo9}))
$$ L_{n}\Big(x,s\,\Big|\,q^{-1}\Big)\,\equiv\,\sum_{k=0}^{\lfloor\,n/2\,
\rfloor}\,c_{n,\,k}^{(L)}(q^{-1})\,s^{\,k}\,x^{\,n-\,2\,k}= \sum_{k=0}
^{\lfloor\,n/2\,\rfloor}\,q^{k(k-n)}\,c_{n,\,k}^{(L)}(q)\,s^{\,k}\,
x^{\,n-\,2\,k}                                                   $$
\be =x^{\,n}{}_4\phi_3\Bigg(q^{-\,n/2},q^{(1-n)/2},-\,q^{-\,n/2},-\,q^
{(1-n)/2}\,;q^{1-\,n}\,,0\,,0\Bigg|\,q\,;-\frac{q\,s}{x^{\,2}}\Bigg)\,. \ee
The next step is to take into account the Fourier integral transform
(4.6) in order to prove the following theorem.
\begin{theorem}
The Fourier integral transform of the $q$-Lucas polynomials $L_{n}(b\,
e^{\,{\rm i}\,\kappa\,x},s\,|\,q)$ times the Gauss exponential function
$e^{-\,x^2/2}$ has the form:
\be
\int_{\RR}\,L_{n}\Big(b\,e^{\,{\rm i}\,\kappa\,x},s\,\Big|\,q\Big)
\,e^{\,{\rm i}\,x\,y\,-\,x^2/2}\,dx\,=\,\sqrt{2\pi}\,q^{\,n^2/4}\,
L_{n}\Big(b\,e^{-\,\kappa\,y},s\,\Big|\,q^{-1}\Big)\,e^{-\,y^2/2}
\,,                                                           \ee
where $b$ is an arbitrary constant factor.
\end{theorem}
\begin{proof}
Starting from (\ref{pippo8}) for the $q$-Lucas polynomials,
 using the Fourier integral transform (4.6)  and employing  the interrelation (4.11) between
the coefficients $c_{n,\,k}^{(L)}(q)$ and $c_{n,\,k}^{(L)}(q^{-1})$, we get

\bea
& & \int_{\RR}\,L_{n}\Big(b\,e^{\,{\rm i}\kappa\,x},s\,\Big|\,q\Big)
\,e^{\,{\rm i}\,x\,y\,-\,x^2/2}\,dx                   \nonumber   \\
& &=\sum_{k=0}^{\lfloor\,n/2\,\rfloor}\,c_{n,\,k}^{(L)}(q)\,s^{\,k}\,
b^{\,n-\,2\,k}\,\int_{\RR}\,e^{\,{\rm i}\,x\,y\,+ \,{\rm i}\,(n-\,2k)
\,\kappa\,x-\,x^2/2}\,dx                                 \nonumber \\
& & =\,\sqrt{2\pi}\,q^{\,n^2/4}\,e^{-\,y^2/2}\,\sum_{k=0}^{\lfloor
\,n/2\,\rfloor}\,q^{k(k-n)}\,c_{n,\,k}^{(L)}(q)\,s^{\,k}\,\Big
(b\,e^{-\,\kappa\,y}\Big)^{\,n-\,2\,k}             \nonumber \\
& & =\,\sqrt{2\pi}\,q^{\,n^2/4}\,e^{-\,y^2/2}\,\sum_{k=0}^{\lfloor
\,n/2\,\rfloor}\,c_{n,\,k}^{(L)}(q^{-1})\,s^{\,k}\,\Big(b\,
e^{-\,\kappa\,y}\Big)^{\,n-\,2\,k}             \nonumber \\
& &=\,\sqrt{2\pi}\,q^{\,n^2/4}\,L_{n}\Big(b\,e^{-\,\kappa\,y}
,s\,\Big|\,q^{-1}\Big)\,e^{-\,y^2/2}\,.        \nonumber \eea
\end{proof}

%\setcounter{equation}{0}

%%%%%%%%%%%%%%%%%%%%%%%%%
%%%%%%%%%%%%%%%%%%%%%%%%%
\section{Concluding remarks and outlook}
%%%%%%%%%%%%%%%%%%%%%%%%%
%%%%%%%%%%%%%%%%%%%%%%%%%

We have studied in detail the transformation properties with respect to
Fourier transform of the $q$-Fibonacci and $q$-Lucas polynomials, which
are governed by the non-conventional three-term recurrences (3.1). In particular, we have proved that these families of
$q$-polynomials exhibit a simple transformation behavior (4.7) and (4.13)
under  the classical Fourier integral transform.

Let us emphasize that one actually may use the Mehta--Dahlquist--Matveev
techniques (see \cite{Mehta, Dahl, Matv, Ata, NMARueW07}) in order to show
that the Fourier integral transformation formulas (4.7) and (4.13) in fact entail
a similar behavior of the $q$-Fibonacci and $q$-Lucas polynomials under the
discrete (finite) Fourier transform as well.

It is worthwhile to mention here  that there are other  possibilities
(than Cigler's $F_{n}(x,s|\,q)$ and $L_{n}(x,s|\,q)$) for constructing
$q$-extensions  of Fibonacci and Lucas polynomials of inte\-rest. For
instance, two monic $q$-polynomial families
\be
r_{n}^{(F)}(x|\,q)\,=\,x^{\,n}\,{}_2\phi_1\Bigg(q^{-\,n},
q^{1-n};\,q^{-\,2n}\,\Bigg|\,q^{\,2}\,;-\frac{1}{q x^2}\Bigg)\,, \ee
\be r_{n}^{(L)}(x|\,q)\,=\,x^{\,n}\,{}_2\phi_1\Bigg(q^{-\,n},
q^{1-n};\,q^{\,2(1-n)}\,\Bigg|\,q^2\,;-\frac{q}{x^2}\Bigg)\,,     \ee
represent very natural extensions of the Fibonacci and Lucas polynomials
$p_{\,n}^{(F)}(x)$ and $p_{\,n}^{(L)}(x)$, defined in (2.4) and (2.16)
respectively. Contrary to $F_{n}(x,s|\,q)$ and $L_{n}(x,s|\,q)$, these
$q$-polynomial families do satisfy standard three-term recurrence
relations of the form
\be
r_{n+1}(x|\,q)\,=\,x\,r_{n}(x|\,q)\,+ \frac
{q^{n-1}}{(1+q^n)(1+q^{n+1})}\,r_{n-1}(x|\,q)\,,\ee
where by $r_{n}$ we mean either $r_{n}^{(F)}$ or $r_{n}^{(L)}$.
Moreover, the  $q$-polynomials $r_{n}^{(F)}(x|\,q)$ and  $r_{n}^{(L)}(x|\,q)$
are associated with  the monic $q$-polynomial families
\be \label{5.5}
s_{n}^{(U)}(x|\,q)={\rm i}^{-n}\,r_{n}^{(F)}({\rm i}\,x|\,q)\,,\quad
\quad s_{n}^{(T)}(x|\,q)={\rm i}^{-n}\,r_{n}^{(L)}({\rm i}\,x|\,q)\,,
\ee
which do satisfy the conditions of Favard's characterization
theorem; thus they can be viewed as natural $q$-extensions of the Chebyshev
polynomials $U_n(x)$ and $T_n(x)$. It should be noted that the polynomials
$s_{n}^{(U)}(x|\,q)$ and  $s_{n}^{(T)}(x|\,q)$, explicitly given by
\be \label{5.6}
s_{n}^{(U)}(x|\,q)\,=\,x^{\,n}\,{}_2\phi_1\Bigg(q^{-\,n},
q^{1-n};\,q^{-\,2n}\,\Bigg|\,q^{\,2}\,;\frac{1}{q x^2}\Bigg)\,, \ee
\be s_{n}^{(T)}(x|\,q)\,=\,x^{\,n}\,{}_2\phi_1\Bigg(q^{-\,n},
q^{1-n};\,q^{\,2(1-n)}\,\Bigg|\,q^2\,;\frac{q}{x^2}\Bigg)\,,   \label{5.6}  \ee
are of interest on their own. 

 It is well known that the Chebyshev polynomials
$U_n(x)$ and $T_n(x)$ are the special cases of the Jacobi polynomials
$P_{n}^{(\alpha,\beta)}(x)$ with the parameters $\alpha=\beta=1/2$ and
$\alpha=\beta=-1/2$  respectively. Therefore it seems natural to expect that
the continuous $q$-Jacobi polynomials $P_{n}^{(\alpha,\beta)}(x|\,q)$ (which
evidently represent  $q$-extensions of the Jacobi polynomials $P_{n}^{(\alpha,\beta)}(x)$)
with the particular values the parameters $\alpha=\beta=1/2$ and $\alpha=\beta=-1/2$,
could provide appropriate $q$-extensions of the Chebyshev polynomials $U_n(x)$
and $T_n(x)$  respectively. Under closer examination however, it turns out
that the continuous $q$-Jacobi polynomials $P_{n}^{(1/2,1/2)}(x|\,q)$ and
$P_{n}^{(-1/2,-1/2)}(x|\,q)$ {\it are only constant (but $q$-dependent) multiples}
of the Chebyshev polynomials $U_n(x)$ and $T_n(x)$. In other words, the continuous
$q$-Jacobi polynomials $P_{n}^{(1/2,1/2)}(x|\,q)$ and $P_{n}^{(-1/2,-1/2)}(x|\,q)$
differ from the Chebyshev polynomials $U_n(x)$ and $T_n(x)$  only for the choice of
the normalization constants; therefore the former two polynomial families are just trivial
$q$-extensions of the latter ones.\footnote{This curious ``$q$-degeneracy" of the
continuous $q$-Jacobi polynomials $P_{n}^{(\alpha,\beta)}(x|\,q)$ for the values of
the parameters $\alpha=\beta=1/2$ and $\alpha=\beta=-1/2$  was first noticed by
R.Askey and J.A.Wilson in their seminal work \cite{AW}. We are grateful to Tom
Koornwinder for reminding us of this fact.}

The polynomials $ s_{n}^{(U)}(x|\,q)$ and $ s_{n}^{(T)}(x|\,q)$ can be thus viewed as
non-trivial compact $q$-extensions of the Chebyshev polynomials $U_n(x)$ and $T_n(x)$,
which do not match with the continuous $q$-Jacobi polynomials $P_{n}^{(1/2,1/2)}(x|\,q)$
and $P_{n}^{(-1/2,-1/2)}(x|\,q)$. Observe that both of them can be expressed in terms
of the little $q$-Jacobi polynomials $p_{n}(x;a,b|\,q)$ as
$$
s_{2n}^{(U)}(x|\,q)\,=\,(-1)^n\,q^{n(n-1)}\frac{(q;q^2)_n}{(q^{2(n+1)};q^2)_{n}}
\,p_{n}\Big(x^2;q^{-1},q\Big|\,q^2\Big)\,,
$$
\be \label{5.8}
s_{2n+1}^{(U)}(x|\,q)\,=\,(-1)^n\,q^{n(n-1)}\frac{(q^3;q^2)_n}{(q^{2(n+2)};q^2)_{n}}
\,x\, p_{n}\Big(x^2;q,q\Big|\,q^2\Big)\,,
\ee
and
$$
s_{2n}^{(T)}(x|\,q)\,=\,(-1)^n\,q^{n(n-1)}\frac{(q;q^2)_n}{(q^{2n};q^2)_{n}}
\,p_{n}\Big(x^2;q^{-1},q^{-1}\Big|\,q^2\Big)\,,
$$
\be \label{5.9}
s_{2n+1}^{(T)}(x|\,q)\,=\,(-1)^n\,q^{n(n-1)}\frac{(q^3;q^2)_n}{(q^{2(n+1)};q^2)_{n}}
\,x\, p_{n}\Big(x^2;q,q^{-1}\Big|\,q^2\Big)\,,
\ee
where $p_{n}(x;a,b\,|\,q):={}_2\phi_1(q^{- n}, ab\,q^{n+1}; aq\,|\,q\,; qx)$
(see, for example, (14.12.1), p.482 in \cite{KSL}). It is of considerable
interest to examine the properties of the polynomials $ s_{n}^{(U)}(x|\,q)$
and $ s_{n}^{(T)}(x|\,q)$ in more detail, including their transformation
properties with respect to the Fourier integral transform. This project is
beyond the subject of this paper and will be dealt with elsewhere.

%%%%%%%%%%%%%%%%%%%%%%%%
%%%%%%%%%%%%%%%%%%%%%%%%
\section*{Acknowledgements}
%%%%%%%%%%%%%%%%%%%%%%%%
%%%%%%%%%%%%%%%%%%%%%%%%

Discussions with T.H.Koornwinder and K.B.Wolf are gratefully acknowledged. 
One of us (NMA) would like to thank the Physics Department and Department
of Electronic Engineering, University Roma Tre, Italy, for the hospitality 
extended to him during his visit in September--October 2011, when the final 
part of this work was carried out. The participation of NMA in this work has 
been supported by the DGAPA-UNAM IN105008-3 and SEP-CONACYT 79899 projects 
``\,\'Optica Matem\'atica". DL  and OR have been partly supported by the 
Italian Ministry of Education and Research, 2010 PRIN ``Continuous and discrete 
nonlinear integrable evolutions: from water waves to symplectic maps".

\bigskip
\noindent {\Large{\bf Appendices}}

\appendix
%%%%%%%%%%%%%%%%%%%%%%
%%%%%%%%%%%%%%%%%%%%%%
\section{Proof of (2.10)}
%%%%%%%%%%%%%%%%%%%%%%
%%%%%%%%%%%%%%%%%%%%%%

 In order to give a direct proof of the transformation formula (2.10)
%\be (n+1)\,{}_2F_1\Big(-n\,,n+2\,;\,3/2\,\Big|\frac{1-z}{2}\, \Big)\,=\,(2 z)^{\,n}\,{}_2F_1\Big(-\frac{n}{2}\,, \frac{1-n}{2}\,;-\,n\,\Big|z^{-\,2}\,\Big)\,,             \ee
 we start by the defining
relation for the hypergeometric ${}_2F_1$-polynomial on the
left side of (2.10) and rewrite it
\bea
\!\!\!{}_2F_1\Big(-n\,,n+2\,;\,3/2\,\Big|\frac{1-z}{2}\,\Big)\!
&\!:=& \sum_{k=0}^n \frac{(-\,n)_k (n+2)_k}{(3/2)_k}\frac
{(1-z)^k}{2^k\,k!}                          \nonumber   \\
\!&=&\sum_{k=0}^n {n\atopwithdelims()k}\frac{(-1)^k\,(n+2)_k}{2^k
\,(3/2)_k}\,\sum_{l=0}^k\,{k\atopwithdelims()l}(-z)^l,     \eea
by employing the relation $(-n)_k=(-1)^k\,n!/(n-k)!$. The next
step is to reverse the order of summation in (A.1) with respect
to the indices $k$ and $l$, which leads to the relation
\newpage
$$
{}_2F_1\Big(-n\,,n+2\,;\,3/2\,\Big|\frac{1-z}{2}\,\Big)$$
\be
=\frac{\Gamma(3/2)}{n+1}\sum_{l=0}^n \frac{\Gamma(2n+2-l)}
{\Gamma(n-l+3/2)}\frac{(z/2)^{n-\,l}}{l!\,(n-l)!}\sum_{k=0}^l
\frac{(-l)_k\,(2n+2-l)_k}{(n-l+3/2)_k}\frac{1}{2^l\,l!}\,. \ee
The sum over index $k$ in (A.2) represents the hypergeometric
polynomial
$$
{}_2F_1\Big(-l\,,2n+2-l\,;\,n-l+3/2\,\Big|\,x\,\Big)  $$
for a special value of the variable $x=1/2$, which can be
evaluated by Gauss's second summation theorem (see, for example,
(1.7.1.9) on p.32 in \cite{Slat})
\bea
{}_2F_1\Big(2a\,,2b\,;\,a+b+1/2\,\Big|\,1/2\,\Big)&=&\frac
{\Gamma\Big(1/2\Big)\Gamma\Big(a+b+1/2\Big)}{\Gamma\Big(a+
1/2\Big)\Gamma\Big(b+1/2\Big)}\,,             \nonumber \\
a+b+1/2&\neq&-\,m\,,\quad\quad m\geq0\,,              \eea
with $a=-\,l/2$ and $b=n+1-l/2$ (so that $a+b+1/2=n-l+3/2\geq3/2$
for all $0\leq l\leq n$). The sum over index $k$ in (A.2) thus
reduces to
\be
{}_2F_1\Big(-l\,,2n+2-l\,;\,n-l+3/2\,\Big|\,1/2\,\Big)
=\frac{\Gamma\Big(1/2\Big)\Gamma\Big(n-l+3/2\Big)}{\Gamma
\Big((2n-l+3)/2\Big)\Gamma\Big((1-l)/2\Big)}\,.       \ee
Since the gamma function $\Gamma(z)$ has poles at the points $z=-\,n$,
$\,n\geq 0$, the right-hand side of (6.5) vanishes for all odd
values of the index $l$ due to the presence of the factor
$\Gamma\Big((1-l)/2\Big)$ in its denominator. This means that
only terms with even $l$'s give non-zero contribution into the
sum over $l$ in (A.2), that is,
\bea
& & {}_2F_1\Big(-n\,,n+2\,;\,3/2\,\Big|\frac{1-z}
{2}\,\Big)                            \nonumber \\
& &=\frac{\pi}{2(n+1)}\sum_{m=0}^{\lfloor\,n/2\,\rfloor}
\frac{\Gamma(2n+2-2m)}{\Gamma(n-m+3/2)\Gamma(1/2-m)}
\,\frac{(z/2)^{n-2m}}{(n-2m)!\,(2m)!}  \nonumber \\
& &=\frac{\sqrt{\pi}}{(n+1)}\sum_{m=0}^{\lfloor\,n/2\,\rfloor}
\frac{\Gamma(n+1-m)}{\Gamma\Big(\frac{n+1}{2}-m\Big)\Gamma\Big
(\frac{n}{2}+1-m\Big)}\,\frac{(-1)^m\,z^{n-2m}}{m!}\,,    \eea
where at the last step we employed duplication formula
\be
\Gamma(2z)\,=\,\frac{2^{\,2z-1}}{\sqrt{\pi}}\,\Gamma(z)
\,\Gamma(z+1/2)                                    \ee
for the gamma function $\Gamma(z)$ and the relation $\Gamma(m+1/2)
\,\Gamma(1/2-m)=\pi/\cos{m\pi}=(-1)^m\,\pi$. Finally, it remains
only to use the identity $\Gamma(z+1-n)=(-1)^n\,\Gamma(z+1)/\,(-z)_n$
and again the duplication formula (A.6) in order to show that
\bea
& & {}_2F_1\Big(-n\,,n+2\,;\,3/2\,\Big|\frac{1-z}{2}
\,\Big)                                  \nonumber \\
& &=\frac{\sqrt{\pi}}{n+1}\,\frac{\Gamma(n+1)}{\Gamma
\Big(\frac{n+1}{2}\Big)\Gamma\Big(\frac{n}{2}+1\Big)}
\,\sum_{m=0}^{\lfloor\,n/2\,\rfloor}\frac{\Big(\!-\frac
{n}{2}\Big)_m\Big(\frac{1-n}{2}\Big)_m}{(-n)_m}\,\frac
{z^{n-2m}}{m!}                            \nonumber \\
& &= \frac{2^n}{n+1}\,\sum_{m=0}^{\lfloor\,n/2\,\rfloor}
\frac{\Big(\!-\frac{n}{2}\Big)_m\Big(\frac{1-n}{2}\Big)
_m}{(-n)_m}\,\frac{z^{n-2m}}{m!}          \nonumber \\
& &=\frac{(2z)^n}{n+1}\,\,{}_2F_1\Big(-\frac{n}{2}\,,
\frac{1-n}{2}\,;\,-n\,\Big|\,z^{-\,2}\,\Big)\,.    \eea
This completes the proof of transformation formula (2.10).

%%%%%%%%%%%%%%%%%%%%%%%
%%%%%%%%%%%%%%%%%%%%%%%
\section{Proof of (2.21)}
%%%%%%%%%%%%%%%%%%%%%%%
%%%%%%%%%%%%%%%%%%%%%%%
We  prove here the  transformation
formula (2.21)
%\be {}_2F_1\Big(-n\,,n\,;\,1/2\,\Big|\frac{1-z}{2}\,\Big)\,=\,2^{n-1}\,z^{\,n}\,{}_2F_1\Big(-\frac{n}{2}\,,\frac{1-n}{2}\,;1-\,n\,\Big|1/z^{2}\,\Big)\,,\quad n\geq1\,,\ee
which was stated in section 2.  One starts with the defining
relation for the hypergeometric ${}_2F_1$ polynomial on the
left side of (2.21) and evaluates first (for $n\geq1$) that
\bea
\!\!\!{}_2F_1\Big(-n\,,n\,;\,1/2\,\Big|\frac{1-z}{2}
\,\Big)\!&\!:=& \sum_{k=0}^n \frac{(-\,n)_k\,(n)_k}
{(1/2)_k}\frac{(1-z)^k}{2^k\,k!}      \nonumber  \\
\!&=&\sum_{k=0}^n {n\atopwithdelims()k}\frac{(-1)^k\,(n)_k}
{2^k\,(1/2)_k}\,\sum_{l=0}^k\,{k\atopwithdelims()l}(-z)^l, \eea
by employing the relation $(-n)_k=(-1)^k\,n!/(n-k)!$. The next step
is to reverse the order of summation in (B.1) with respect to the
indices $k$ and $l$, which leads to the relation
$$
{}_2F_1\Big(-n\,,n\,;\,1/2\,\Big|\frac{1-z}{2}\,\Big)$$
\be
=\,n\,\Gamma(1/2)\,\sum_{l=0}^n \frac{\Gamma(2n-l)}{\Gamma
(n-l+1/2)}\frac{(z/2)^{n-\,l}}{l!\,(n-l)!}\sum_{k=0}^l\frac
{(-l)_k\,(2n-l)_k}{(n-l+1/2)_k}\frac{1}{2^k\,k!}\,.     \ee
The sum over index $k$ in (B.2) represents the hypergeometric
polynomial
$$
{}_2F_1\Big(-l\,,2n-l\,;\,n-l+1/2\,\Big|\,x\,\Big)  $$
for a special value of the variable $x=1/2$, which can be
evaluated by Gauss's second summation theorem (see, for example,
(1.7.1.9) on p.32 in \cite{Slat})
\bea
{}_2F_1\Big(2a\,,2b\,;\,a+b+1/2\,\Big|\,1/2\,\Big)&=&\frac
{\Gamma\Big(1/2\Big)\Gamma\Big(a+b+1/2\Big)}{\Gamma\Big(a+
1/2\Big)\Gamma\Big(b+1/2\Big)}\,,             \nonumber \\
a+b+1/2&\neq&-\,m\,,\quad\quad m\geq0\,,              \eea
with parameters $a=-\,l/2$ and $b=n-l/2$ (so that $a+b+1/2=n-l+1/2
\geq1/2$ for all $0\leq l\leq n$). The sum over index $k$ in (B.2)
thus reduces to
\be
{}_2F_1\Big(-l\,,2n-l\,;\,n-l+1/2\,\Big|\,1/2\,\Big)
=\frac{\Gamma\Big(1/2\Big)\Gamma\Big(n-l+1/2\Big)}{\Gamma
\Big((2n-l+1)/2\Big)\Gamma\Big((1-l)/2\Big)}\,.       \ee
Since the gamma function $\Gamma(z)$ has poles at the points $z=-\,n$,
$\,n\geq 0$, the right-hand side of (B.4) vanishes for all odd values of
the index $l$ due to the presence of the factor $\Gamma\Big((1-l)/2\Big)$
in its denominator. This means that only terms with even $l$'s give
non-zero contribution into the sum over $l$ in (B.2), that is,
\bea
& & {}_2F_1\Big(-n\,,n\,;\,1/2\,\Big|\frac{1-z}
{2}\,\Big)                         \nonumber \\
& &=n\,\pi\,\sum_{m=0}^{\lfloor\,n/2\,\rfloor}\frac
{\Gamma(2n-2m)}{\Gamma(n-m+1/2)\,\Gamma(1/2-m)}\,\frac
{(z/2)^{\,n-\,2m}}{(n-2m)!\,(2m)!}  \nonumber \\
& &=\frac{n}{2}\,\sqrt{\pi}\,\sum_{m=0}^{\lfloor\,n/2\,\rfloor}
\frac{\Gamma(n-m)}{\Gamma\Big(\frac{n+1}{2}-m\Big)\Gamma\Big
(\frac{n}{2}+1-m\Big)}\,\frac{(-1)^m\,z^{\,n-\,2m}}{m!}\,, \eea
where at the last step we employed duplication formula (A.6) for
the gamma function $\Gamma(z)$ and the relation $\Gamma(m+1/2)
\,\Gamma(1/2-m)=\pi/\cos{m\pi}=(-1)^m\,\pi$. Finally, it remains
only to use the identity $\Gamma(z+1-n)=(-1)^n\,\Gamma(z+1)/
\,(-z)_n$ and again the duplication formula (A.6) in order to
show that

\bea
& & {}_2F_1\Bigg(-n\,,n\,;\,1/2\,\Bigg|\frac{1-z}{2}
\,\Bigg)                                \nonumber \\
& &=\frac{\sqrt{\pi}}{2}\,\frac{\Gamma(n+1)}{\Gamma
\Big(\frac{n+1}{2}\Big)\Gamma\Big(\frac{n}{2}+1\Big)}
\,\sum_{m=0}^{\lfloor\,n/2\,\rfloor}\frac{\Big
(\!-\frac{n}{2}\Big)_m\Big(\frac{1-n}{2}\Big)_m}
{(1-n)_m}\,\frac{z^{\,n-\,2m}}{m!}  \nonumber \\
& &= 2^{n-1}\,\sum_{m=0}^{\lfloor\,n/2\,\rfloor}
\frac{\Big(\!-\frac{n}{2}\Big)_m\Big(\frac{1-n}{2}
\Big)_m}{(1-\,n)_m}\,\frac{z^{\,n-\,2m}}{m!}\nonumber\\
& &=2^{n-1}\,z^n\,{}_2F_1\Big(-\frac{n}{2}\,,\frac
{1-n}{2}\,;\,1-n\,\Big|\,1/z^{\,2}\,\Big)\,.  \eea
This completes the proof of the transformation formula (2.21).

\end{document}